\documentclass[onecolumn,draftcls,12pt]{IEEEtran}
\usepackage{amssymb}
\usepackage{graphicx}
\usepackage{amsmath}
\usepackage{amsfonts}
\usepackage{latexsym}
\usepackage{bm}
\usepackage{cite}
\usepackage{url}
\makeatletter
\def\hlinew#1{%
  \noalign{\ifnum0=`}\fi\hrule \@height #1 \futurelet
   \reserved@a\@xhline}
\ifCLASSINFOpdf
\else
\fi

\begin{document}

\title{Energy Efficiency Scaling Law for MIMO Broadcasting Channels\thanks{This work is supported by National Basic Research Program of China (973 Program)
2007CB310602.}\thanks{This work has been submitted to the IEEE for possible publication.  Copyright may be transferred without notice, after which this version may no longer be accessible.}}

\author{Jie~Xu and Ling~Qiu
\thanks{The authors are with the Personal Communication Network \& Spread Spectrum Laboratory,
Department of Electrical Engineering and Information Science,
University of Science and Technology of China Hefei, Anhui, 230027,
China (email: suming@mail.ustc.edu.cn, lqiu@ustc.edu.cn).}
\thanks{Corresponding author: Ling Qiu, lqiu@ustc.edu.cn.}}

\markboth{IEEE Wireless Communication Letters}{Jie Xu \lowercase{\textit{et al}}.: Energy Efficiency Scaling Law for MIMO Broadcasting Channels}

\maketitle

\begin{abstract}
This letter investigates the energy efficiency (EE) scaling law for the broadcasting channels (BC) with many users, in which the non-ideal transmit independent power consumption is taken into account. We first consider the single antenna case with $K$ users, and derive that the EE scales as $\frac{{\log_2 \ln K}}{ \alpha}$
  when $\alpha > 0$  and  $\log_2 K$ when  $\alpha = 0$, where $\alpha$ is the normalized transmit independent power. After that, we extend it to the general MIMO BC case with a $M$-antenna transmitter and $K$ users each with $N$ antennas. The scaling law becomes
$\frac{{M \log_2 \ln NK}}{ \alpha}$  when   $\alpha > 0$ and $ \log_2 NK$  when  $\alpha = 0$.
\end{abstract}


%
\IEEEpeerreviewmaketitle

\section{Introduction}
Energy efficiency (EE) is becoming increasingly important for the future green wireless communication systems \cite{YChenComMag}. As multiuser multiple input multiple output (MIMO) is the key technology for the next generation cellular networks, understanding the EE scaling behavior of the MIMO broadcasting channels (BC) is a critical issue to help the design of the green wireless networks. Therefore, we will investigate the scaling law of the EE for the MIMO BC in this letter.

The capacity scaling law of the MIMO BC is a well studied topic. It is known that the dirty paper coding (DPC) is the capacity achieving scheme \cite{DPC}. With DPC, the capacity scaling law of MIMO BC has been widely investigated in the literature, and the  $M \log_2 \ln NK$ scaling behavior is famous \cite{Scaling_TSDPC,Limits_mumimo}. Surprisingly, some sub-optimal schemes with low complexity precoding and user selection \cite{ZFDPC,ZFBFSUS,Scaling_RBF} can also achieve this scaling law, which makes them promising in the real systems. However, to the best of our knowledge, the EE scaling law has not been addressed yet, and thus is unknown.

EE is in general defined as the capacity divided by the total power consumption including the transmit-dependent and independent power, which represents the delivered bits per unit energy, measured in bits per Joule. To optimize the EE, adjusting the transmit power is the basis and the fractional programming is always utilized as the mathematical tool \cite{fractionalprogramming}. Based on this tool, an energy efficient iterative waterfilling scheme is shown to be optimal for the EE of the MIMO BC \cite{EEWaterfillingDPC}. The distinct feature of optimizing EE is that the sum transmit power should be adjusted according to the channel conditions and the transmit-independent power. Therefore,  compared with the derivation of the capacity scaling law which is based on fixed transmit power, this power adjusting feature would make the EE scaling law different.

The scaling law of the MIMO BC is first investigated in this letter, with the help of the Lambert $\omega$ function \cite{LambertWfunction}. The main contribution is as follows. We first derive that the EE scales as  $\frac{{\log_2 \ln K}}{ \alpha}$
when $\alpha > 0$  and as  $\log_2 K$ when  $\alpha = 0$ in the SISO scenario. After that, we extend it to the general multi-antenna case and the scaling law becomes $\frac{{M \log_2 \ln NK}}{ \alpha}$  when   $\alpha > 0$ and  $\log_2 NK$  when  $\alpha = 0$. The results give us insights about the effect of parameters, such as user number, transmit antenna number, on the EE,  and would help the design of the future green wireless networks.

The rest of this letter is organized as follows. Section \ref{sec2} introduces the system and power model. The scaling laws for the SISO and MIMO cases are given in section \ref{sec3} and \ref{sec4} respectively. Finally, section \ref{sec5} concludes this letter.

Regarding the notation, bold face letters refer to vectors (lower case) or matrices (upper case). The superscript $H$ and $T$ represent the
conjugate transpose and transpose operation, respectively. ${\rm tr}(\cdot)$ denotes the trace of the matrix,  and ${\mathbb E}(\cdot)$ denotes the expectations of random variables.

\section{System Model}\label{sec2}
The system consists of a base station (BS) with $M$ transmit antennas and $K$ users each with $N$ receive antennas. We consider the homogeneous scenario and assume that each user has the same large scale fading including pathloss and shadowing, which can be denoted as $\psi$. The smaller scale Rayleigh fading is considered and that from the BS to the $k$th user ${\bf H}_k \in {\mathbb{C}}^{N \times M}$ is  a zero-mean Gaussian random matrix, with each entry independent and identically distributed (i.i.d.) ${\mathcal{CN}}(0,1)$. Assume the transmitted signal at the BS is ${\bf{s}}\in {\mathbb{C}}^{M \times 1}$, and then the received signal of user $k$ can be denoted as
\begin{equation} \label{sys1}
\begin{array}{l}
{\bf y}_k = \psi{\bf H}_k {\bf{s}} + {\bf{n}},
\end{array}
\end{equation}
where ${\bf n}$ is the noise at the user. Denote the transmit power as $P = {\rm tr}({\bf{s}}{\bf{s}}^H)$. Based on the uplink-downlink duality \cite{Duality}, the sum capacity of the MIMO BC can be denoted as
\begin{equation} \label{sys2}
\begin{array}{l}
\displaystyle C(P) = \\ \displaystyle \max \limits_{{\bf{Q}}_k \ge 0, \sum \limits_{k=1}^{K}{\rm tr}({\bf Q}_k)\le P} W \log_2 \left| {{\bf{I}} + \frac{\psi}{N_0W}\sum\limits_{k = 1}^K {{\bf{H}}_k^H{{\bf{Q}}_k}{{\bf{H}}_k}} }\right|,
\end{array}
\end{equation}
where ${\bf{Q}}_k$ is the transmit covariance of the users in the dual uplink, $W$ is the bandwidth and $N_0$ is the density of the noise power.

As the BS takes the main parts of the total power consumption in the cellular networks \cite{Arnold}, we only consider the power consumption at the BS in this letter. Based on the realistic BS model \cite{Arnold} and our previous work \cite{EEWaterfillingDPC,Xu1}, the power model is denoted as
\begin{equation} \label{sys3}
\begin{array}{l}
\displaystyle P_{\rm{total}}(P) = \frac{P}{\eta} + M P_{\rm{dyn}} + P_{\rm{sta}},
\end{array}
\end{equation}
where $\eta$ denotes the power amplifier (PA) efficiency; $M P_{\rm{dyn}}$ denotes the dynamic power consumption proportional to the number of
radio frequency (RF) chains, e.g. circuit power of RF chains which is always proportional to $M$; and $P_{\rm{sta}}$ accounts for the static
power independent of both $M$ and $P$ which includes power consumption of the baseband processing, battery unit etc.. $M P_{\rm{dyn}} + P_{\rm{sta}}$ is referred to as the transmit-independent power.

Therefore, the maximum EE of the MIMO BC can be defined as
\begin{equation} \label{sys4}
\begin{array}{l}
\displaystyle \Gamma = \max \limits_{P>0} \frac{C(P)}{P_{\rm{total}}(P)}.
\end{array}
\end{equation}
To simplify the notation, we define that $\frac{\psi}{N_0W}P = Q$, $\alpha = {\frac{{{N_0}W}\eta}{\psi }(M{P_{{\rm{dyn}}}} + {P_{{\rm{sta}}}}})$, and then
\begin{equation} \label{sys5}
\begin{array}{l}
\displaystyle \Gamma =  \frac{{W^2{N_0}\eta}}{\psi } \xi,
\end{array}
\end{equation}
where
\begin{equation} \label{sys6}
\begin{array}{l}
\displaystyle \xi=\mathop {\max }\limits_{Q \ge 0} \frac{{\mathop {\max }\limits_{{{\bf{Q}}_k} \ge 0,\sum\limits_{k = 1}^K {{\rm{tr}}} ({{\bf{Q}}_k}) \le Q} \log_2 \left| {{\bf{I}} + \sum\limits_{k = 1}^K {{\bf{H}}_k^H{{\bf{Q}}_k}{{\bf{H}}_k}} } \right|}}{{ {Q + \alpha } }}.
\end{array}
\end{equation}
As our purpose is to derive the scaling law of the EE of the MIMO BC, and $\alpha$, $\frac{{W^2{N_0}\eta}}{\psi }$ can be viewed as constant here, we would utilize $\xi$ as the normalized EE metric in this letter. And $\alpha$ is viewed as the normalized transmit-independent power.

Note that the solution of (\ref{sys6}) can be obtained based on the energy efficient iterative waterfilling in \cite{EEWaterfillingDPC}, however, it is difficult to derive the scaling law directly from \cite{EEWaterfillingDPC}. Before discussing the general MIMO scenario, let us look at the SISO case at first.

\section{Scaling Law for the SISO Case}\label{sec3}

For the SISO case, only transmitting to the user with the largest channel gain is the optimal solution \cite{Limits_mumimo}, and then the EE can be denoted as
\begin{equation} \label{siso1}
\begin{array}{l}
\displaystyle \xi=\mathop {\max }\limits_{Q \ge 0} \frac{\log_2 (1 + Q \max \limits_{1 \le i \le K} |H_i|^2)}{{ {Q + \alpha } }}.
\end{array}
\end{equation}
Fortunately, with the help of the Lambert $\omega$ function, we can obtain the close-form solution for (\ref{siso1}). We will give the following definition and lemma at first.

\newtheorem{Definition}{Definition}
\begin{Definition}[Lambert $\omega$ function \cite{LambertWfunction}]\label{Definition1}
The Lambert $\omega$ function is defined as the inverse function of \[f(X) = Xe^X,\]
where $X$ is any complex number.
\end{Definition}

\newtheorem{Lemma}{Lemma}
\begin{Lemma}\label{Lemma0}
For the optimization problem
\begin{equation} \label{siso2}
\begin{array}{l}
\displaystyle \mathop {\max }\limits_{x \ge 0} \frac{\log_2 (1 + \gamma x)}{{\left( {x + \alpha } \right)}},
\end{array}
\end{equation}
if we denote that $y = {\log_2 (1 + \gamma x)}$, the optimal solution is given as
\begin{equation} \label{siso3}
\begin{array}{l}
\displaystyle  y^* = \frac{1}{{\ln 2}}\left[ {\omega \left( {\frac{{\alpha \gamma  - 1}}{e}} \right) + 1} \right],
\end{array}
\end{equation}
and the corresponding $x^*$ is
\begin{equation} \label{siso4}
\begin{array}{l}
\displaystyle  x^* = \frac{{\left( {\alpha \gamma  - 1} \right){\omega ^{ - 1}}\left( {\frac{{\alpha \gamma  - 1}}{e}} \right) - 1}}{\gamma }.
\end{array}
\end{equation}
\end{Lemma}

\begin{proof}
Based on \cite{frac_program}, the optimization of (\ref{siso2}) is a quasi-concave optimization, and  any local optimal point is globally optimal. We can denote $y = {\log_2 (1 + \gamma x)}$, and then the optimal $y^*$ is
\begin{equation} \label{siso5}\begin{array}{l}
\displaystyle y^* = \arg \mathop {\max }\limits_{y \ge 0} \frac{y}{{ {\frac{2^y-1}{\gamma} + \alpha } }}\\
=\displaystyle \arg \mathop {\min }\limits_{y \ge 0} \frac{{ {{2^y-1} + \alpha\gamma } }}{\gamma y}.
\end{array}\end{equation}
Calculating the first order derivative, we have that
\begin{equation} \label{siso6}\begin{array}{l}
\displaystyle \left( {\ln 2\cdot y^* - 1} \right) {2^{y^*}} = \alpha \gamma - 1 .
\end{array}\end{equation}
Thus, (\ref{siso3}) can be obtained. Correspondingly, we have that
\[
\begin{array}{l}
\displaystyle  x^* = \frac{{2^{y^*}-1}}{\gamma }\\ \displaystyle  = \frac{\exp\left\{\left[ {\omega \left( {\frac{{\alpha \gamma  - 1}}{e}} \right) + 1} \right]\right\} - 1}{\gamma}.
\end{array}
\]
According to the property of the Lambert $\omega$ function, $\omega(X)e^{\omega(X)} = X$, and then $e^{\omega(X)} = X\omega^{-1}(X)$. Therefore, taking $X = \frac{{\alpha \gamma  - 1}}{e}$, and then (\ref{siso4}) can be obtained.
\end{proof}

Based on Lemma \ref{Lemma0} and treating $Q$, $\max \limits_{1 \le i \le K} |H_i|^2$ as $x$, $\gamma$ respectively, the EE of (\ref{siso1}) can be denoted as
\begin{equation} \label{siso7}
\begin{array}{l}
\displaystyle \xi=\frac{{\frac{1}{{\ln 2}}\left[ {\omega \left( {\frac{{\alpha \mathop {\max }\limits_{1 \le i \le K} |{H_i}{|^2} - 1}}{e}} \right) + 1} \right]}}{{{Q^*} + \alpha }},
\end{array}
\end{equation}
where \begin{equation} \label{siso8}
\begin{array}{l} \displaystyle Q^* = \frac{{\left( {\alpha \max \limits_{1 \le i \le K} |H_i|^2  - 1} \right){\omega ^{ - 1}}\left( {\frac{{\alpha \max \limits_{1 \le i \le K} |H_i|^2  - 1}}{e}} \right) - 1}}{\max \limits_{1 \le i \le K} |H_i|^2 }.
\end{array}
\end{equation}
We have the following lemma.
\begin{Lemma}\label{Lemma1}
When $K \to \infty$, we have that $Q^* \to 0$.
\end{Lemma}
\begin{proof}
When $K \to \infty$, $\max \limits_{1 \le i \le K} |H_i|^2 \to \infty$, and then ${\omega }\left( {\frac{{\alpha \max \limits_{1 \le i \le K} |H_i|^2  - 1}}{e}} \right) \to \infty$. Look at (\ref{siso8}), we can obtain that $Q^* \to 0$ easily.
\end{proof}

We can try to obtain the scaling law as the following Theorem based on Lemma \ref{Lemma1} and (\ref{siso7}).

\newtheorem{Theorem}{Theorem}
\begin{Theorem}\label{Theorem1}
When $K \to \infty$, we have that
\begin{equation} \label{eq11}
\begin{array}{l}
\displaystyle \lim \limits_{K\to \infty} \frac{{\mathbb E}(\xi)}{\frac{ \log_2 \ln K }{\alpha}} =1
\end{array}
\end{equation}
when $\alpha > 0$ and
\begin{equation} \label{eq110}
\begin{array}{l}
\displaystyle \lim \limits_{K\to \infty} \frac{{\mathbb E}(\xi)}{\log_2 K } =1
\end{array}
\end{equation}
when $\alpha = 0$.
\end{Theorem}

\begin{proof}
We will look at the $\alpha > 0$ case at first.

According to Lemma \ref{Lemma1}, we have that
\begin{equation} \label{siso9}
\begin{array}{l}
\displaystyle \xi \approx \frac{{\frac{1}{{\ln 2}}\left[ {\omega \left( {\frac{{\alpha \mathop {\max }\limits_{1 \le i \le K} |{H_i}{|^2} - 1}}{e}} \right) + 1} \right]}}{{ \alpha }}
\end{array}
\end{equation}
when $K \to \infty$. And then  motivated by \cite{Scaling_TSDPC}, as $|H_i|^2$'s have $\chi^2(2)$ distribution, we have that $\beta = {\max \limits_{1 \le i \le K} |H_i|^2} \sim \ln K$. As $\xi$ is a concave function of $\gamma$, following the similar idea in \cite{Scaling_TSDPC}, we can have that
\begin{equation} \label{siso10}
\begin{array}{l}
\displaystyle {\mathbb E}(\xi) \sim \frac{{\frac{1}{{\ln 2}}\left[ {\omega \left( {\frac{{\alpha \ln K - 1}}{e}} \right) + 1} \right]}}{{ \alpha }}
\end{array}
\end{equation}
And then, we have that
\begin{equation} \label{siso11}
\begin{array}{l}
\displaystyle \mathop {\lim }\limits_{K \to \infty } \frac{{\frac{1}{{\ln 2}}\left[ {\omega \left( {\frac{{\alpha \ln K - 1}}{e}} \right) + 1} \right]}}{{\log_2 \ln K}} \\
 \displaystyle \approx \mathop {\lim }\limits_{K \to \infty } \frac{{\omega \left( {\frac{{\alpha \ln K - 1}}{e}} \right)}}{{\ln \ln K}} \\
\displaystyle \mathop  = \limits^{\left( a \right)} \mathop {\lim }\limits_{K \to \infty } \frac{{\frac{{\omega \left( {\frac{{\alpha \ln K - 1}}{e}} \right)}}{{\frac{{\alpha \ln K - 1}}{e}\left[ {\omega \left( {\frac{{\alpha \ln K - 1}}{e}} \right) + 1} \right]}} \cdot \frac{\alpha }{e}\frac{1}{K}}}{{\frac{1}{{K \cdot \ln K}}}} \\
\displaystyle \mathop  =  1,
\end{array}
\end{equation}
where $(a)$ is following the L'Hospital's rule and based on the following property of the first order derivative of the Lambert $\omega$ function \cite{LambertWfunction}.
\[\omega '\left( X \right) = \frac{{\omega \left( X \right)}}{{X\left( {\omega \left( X \right) + 1} \right)}}\]
Therefore, we have that
\begin{equation} \label{siso12}
\begin{array}{l}
\displaystyle {\mathbb E}(\xi) \sim \frac{{{\log_2 \ln K}}}{{ \alpha }},
\end{array}
\end{equation}
and the first part of Theorem \ref{Theorem1} is proved.

When $\alpha = 0$, the maximization of (\ref{siso1}) is achieved when $Q=0$. Based on (\ref{siso1}), we can have that
\begin{equation} \label{siso13}
\begin{array}{l}
\displaystyle {\mathbb{E}}\{\xi\}  = {\mathbb{E}}\left\{\frac{{\mathop {\max }\limits_{1 \le i \le K} |{H_i}{|^2}}}{{\ln 2}}\right\} \sim \frac{{\ln K}}{{\ln 2}} \sim \log_2 K.
\end{array}
\end{equation}
Thus, Theorem \ref{Theorem1} is proved.
\end{proof}

\newtheorem{Remark}{Remark}
Theorem \ref{Theorem1} points out that the scaling law of the EE is affected by $\alpha$ significantly. When $\alpha > 0$, the scaling behavior is similar with the capacity. When $\alpha = 0$, the scaling law becomes different.


\section{Scaling Law for the MIMO Case}\label{sec4}

We will turn to the EE scaling law for the general MIMO BC in this section. As only iterative solution of optimizing the EE of the MIMO BC is available \cite{EEWaterfillingDPC} and obtaining the closed-form expression  is difficult, we would utilize upper and lower bounds to employ Lemma \ref{Lemma0} to  derive the EE scaling law.

The key result in this section is the following theorem.
\begin{Theorem}\label{Theorem2}
When $K \to \infty$, we have that
\begin{equation} \label{eq11}
\begin{array}{l}
\displaystyle \lim \limits_{K\to \infty} \frac{{\mathbb E}(\xi)}{\frac{M \log_2 \ln NK }{\alpha}} =1
\end{array}
\end{equation}
when $\alpha > 0$ and
\begin{equation} \label{eq110}
\begin{array}{l}
\displaystyle \lim \limits_{K\to \infty} \frac{{\mathbb E}(\xi)}{{\log_2 NK }} =1
\end{array}
\end{equation}
when $\alpha = 0$.
\end{Theorem}

\begin{proof}
Based on  \cite{Scaling_TSDPC}, we will give the upper bound at first. For any transmit covariances, we have that
\begin{equation} \label{simu1}\begin{array}{l}
\mathop {\max }\limits_{{{\bf{Q}}_k} \ge 0,\sum\limits_{k = 1}^K {{\rm{tr}}} ({{\bf{Q}}_k}) \le Q} \log_2 \left| {{\bf{I}} + \sum\limits_{i = 1}^K {{\bf{H}}_i^H{{\bf{Q}}_i}{{\bf{H}}_i}} } \right|\\
\displaystyle \mathop
 \le \limits^{\left( b \right)}\mathop {\max }\limits_{{{\bf{Q}}_k} \ge 0,\sum\limits_{k = 1}^K {{\rm{tr}}} ({{\bf{Q}}_k}) \le Q} M\log_2 \left( {1 + \frac{{{\rm{tr}}\left( {\sum\limits_{i = 1}^K {{\bf{H}}_i^H{{\bf{Q}}_i}{{\bf{H}}_i}} } \right)}}{M}} \right) \\
\displaystyle \mathop  \le \limits^{\left( c \right)}\mathop {\max }\limits_{{{\bf{Q}}_k} \ge 0,\sum\limits_{k = 1}^K {{\rm{tr}}} ({{\bf{Q}}_k}) \le Q} M\log_2 \left( {1 + \frac{{\sum\limits_{i = 1}^K {\mathop {\max }\limits_{1 \le j \le N} g_j^ig_j^{iH}{\rm{tr}}\left( {{{\bf Q}_i}} \right)} }}{M}} \right) \\
\displaystyle  = M\log_2 \left( {1 + \frac{Q}{M}\mathop {\max }\limits_{1 \le i \le K} \mathop {\max }\limits_{1 \le j \le N} g_j^ig_j^{iH}} \right)
 \end{array}
 \end{equation}
where $g_j^i$ is the $j$th row of ${\bf H}_i$, (b) follows $\det \left( {\bf{A}} \right) \le {\left( {\frac{{{\rm{tr}}\left( {\bf{A}} \right)}}{M}} \right)^M}$,  (c) follows
${\rm{tr}}\left( { {{\bf{H}}_i^H{{\bf{Q}}_i}{{\bf{H}}_i}} } \right) \le \mathop {\max }\limits_{1 \le j \le N} g_j^ig_j^{iH}{\rm{tr}}\left( {{{\bf Q}_i}} \right)$, and $g_j^ig_j^{iH}$'s are i.i.d random variables with $\chi^2(2M)$ distribution. Thus, the upper bound of $\xi$ can be denoted as
\begin{equation} \label{mimo1}
\begin{array}{l}
\displaystyle \xi^{\rm upp} =  \mathop {\max }\limits_{Q \ge 0} \frac{M\log_2 \left( {1 + \frac{Q}{M}\mathop {\max }\limits_{1 \le i \le K} \mathop {\max }\limits_{1 \le j \le N} g_j^ig_j^{iH}} \right)}{{\left( {Q + \alpha } \right)}}.
\end{array}
\end{equation}
About the lower bound, we will utilize the ZFDPC with greedy scheduling. Based on \cite{ZFDPC}, treating each antenna as a single antenna user\footnote{This is a scheme with the lower bound performance, as the performance can be further improved through cooperation among the antennas within a user.} and assuming equal power allocation, the capacity lower bound can be denoted as
\begin{equation} \label{mimo2}\begin{array}{l}
\mathop {\max }\limits_{{{\bf{Q}}_k} \ge 0,\sum\limits_{k = 1}^K {{\rm{tr}}} ({{\bf{Q}}_k}) \le Q} \log_2 \left| {{\bf{I}} + \sum\limits_{i = 1}^K {{\bf{H}}_i^H{{\bf{Q}}_i}{{\bf{H}}_i}} } \right|\\
\ge \sum\limits_{i = 1}^M {\log_2 \left( {1 + \frac{P}{M}d_{k,k}^2} \right)}  \\
  \ge M\log_2 \left( {1 + \frac{P}{M}d_{M,M}^2} \right), \\
 \end{array}\end{equation}
where $d_{ii}^2$ is the equivalent channel gain of each selected user and the distribution of $d_{kk}^2$ follows $\max \limits_{1 \le i \le NK-k+1} \chi^2\left(2(M-k+1)\right)$. Thus, the lower bound of the EE can be denoted as
 \begin{equation} \label{mimo3}
\begin{array}{l}
\displaystyle \xi^{\rm low} = \mathop {\max }\limits_{Q \ge 0} \frac{M\log_2 \left( {1 + \frac{P}{M}d_{M,M}^2} \right)}{{\left( {Q + \alpha } \right)}}.
\end{array}
\end{equation}

Based on  (\ref{mimo1}),  (\ref{mimo3}), Lemma \ref{Lemma0}, and following the similar procedure in Theorem \ref{Theorem1}, when $\alpha > 0$, we can have the following upper bound and lower bound
\begin{equation} \label{mimo4}
\begin{array}{l}
\displaystyle \xi^{\rm upp} \approx \frac{{\frac{M}{{\ln 2}}\left[ {\omega \left( {\frac{{\alpha \mathop {\max }\limits_{1 \le i \le K} \mathop {\max }\limits_{1 \le j \le N} g_j^ig_j^{iH} - M}}{eM}} \right) + 1} \right]}}{{ \alpha }} \end{array}
\end{equation}
\begin{equation} \label{mimo5}
\begin{array}{l}
\displaystyle \xi^{\rm low} \approx \frac{{\frac{M}{{\ln 2}}\left[ {\omega \left( {\frac{{\alpha d_{M,M}^2 - M}}{eM}} \right) + 1} \right]}}{{ \alpha }} \end{array}
\end{equation}
As $g_j^ig_j^{iH}$'s are i.i.d random variables with $\chi^2(2M)$ distribution, we have that $\mathop {\max }\limits_{1 \le i \le K} \mathop {\max }\limits_{1 \le j \le N} g_j^ig_j^{iH} \sim \ln NK $. Meanwhile, we have $d_{M,M}^2 \sim \ln NK $ as $d_{kk}^2$ follows $\max \limits_{1 \le i \le NK-k+1} \chi^2\left(2(M-k+1)\right)$. Therefore, we have that
\begin{equation} \label{mimo6}
\begin{array}{l}
\displaystyle {\mathbb E}\{\xi\} \le {\mathbb E} \left\{ \frac{{\frac{M}{{\ln 2}}\left[ {\omega \left( {\frac{{\alpha \mathop {\max }\limits_{1 \le i \le K} \mathop {\max }\limits_{1 \le j \le N} g_j^ig_j^{iH} - M}}{eM}} \right) + 1} \right]}}{{ \alpha }} \right\} \\
\displaystyle \sim \frac{{\frac{M}{{\ln 2}}\left[ {\omega \left( {\frac{{\alpha \ln NK - M}}{eM}} \right) + 1} \right]}}{{ \alpha }} \sim \frac{M \log_2 \ln NK }{\alpha}
\end{array}
\end{equation}
\begin{equation} \label{mimo7}
\begin{array}{l}
\displaystyle {\mathbb E}\{\xi\} \ge {\mathbb E} \left\{\frac{{\frac{M}{{\ln 2}}\left[ {\omega \left( {\frac{{\alpha d_{M,M}^2 - M}}{eM}} \right) + 1} \right]}}{{ \alpha }}\right\} \\
\displaystyle \sim \frac{{\frac{M}{{\ln 2}}\left[ {\omega \left( {\frac{{\alpha \ln NK - M}}{eM}} \right) + 1} \right]}}{{ \alpha }} \sim \frac{M \log_2 \ln NK }{\alpha}
\end{array}
\end{equation}
When $\alpha = 0$, following the same procedure in ({\ref{siso13}}), we can have
\begin{equation} \label{mimo8}
\begin{array}{l}
\displaystyle {\mathbb{E}}\{\xi\}  \le {\mathbb{E}}\left\{\frac{\mathop {\max }\limits_{1 \le i \le K} \mathop {\max }\limits_{1 \le j \le N} g_j^ig_j^{iH}}{{\ln 2}}\right\} \sim \frac{{\ln NK}}{{\ln 2}} \sim \log_2 NK.
\end{array}
\end{equation}
and
\begin{equation} \label{mimo9}
\begin{array}{l}
\displaystyle {\mathbb{E}}\{\xi\}  \ge {\mathbb{E}}\left\{\frac{d_{M,M}^2}{{\ln 2}}\right\} \sim \frac{{\ln NK}}{{\ln 2}} \sim \log_2 NK.
\end{array}
\end{equation}
(\ref{eq110}) is proved. Therefore, Theorem \ref{Theorem2} is verified.
\end{proof}

\begin{Remark}
Theorem \ref{Theorem2} gives us insights about the EE of the MIMO BC case. When $\alpha > 0$, it is similar with the SISO case, where the transmit-independent power dominates the denominator. However, things change when $\alpha = 0$. Although the EE scales as $\log_2 NK$, the multiplexing gain is unavailable. That is, varying the transmit antenna number in this case cannot change the EE.

Furthermore, let us look at the effect of the transmit antenna number $M$ when $\alpha > 0$. As $\alpha = {\frac{{{N_0}W}\eta}{\psi }(M{P_{{\rm{dyn}}}} + {P_{{\rm{sta}}}}})$, we can have that
\[{\mathbb E}(\xi ) \sim \frac{\psi }{{{N_0}W\eta }} \cdot \frac{{M\log_2 \ln NK}}{{(M{P_{{\rm{dyn}}}} + {P_{{\rm{sta}}}})}}\]
Therefore, we can conclude that when $K \to \infty$, utilizing more antennas always benefits from the standpoint of both EE and capacity{\footnote{It is worthwhile to note that it does not hold when the user number is limited, in which case there exists a tradeoff between the capacity gain and the increasing power consumption, e.g. see \cite{Xu1}.}}.
\end{Remark}

\section{Conclusions}\label{sec5}

We analyze the EE scaling law for the MIMO BC with many users in this letter. We employ the closed-form solution in the SISO case to derive the EE scaling law at first. After that, we obtain the scaling law  for the general MIMO BC based on the upper and low bound, and borrowing the results of the SISO case. It is shown that the EE scaling law is affected by the transmit-independent power $\alpha$ significantly. When $\alpha > 0$, a scaling law ${\frac{M \log_2 \ln NK }{\alpha}}$ can be acquired, while when $\alpha = 0$, the EE scales  as $\log_2 NK $.

\bibliographystyle{IEEEtran}
\bibliography{reference}

\begin{thebibliography}{10}
\providecommand{\url}[1]{#1}
\csname url@samestyle\endcsname
\providecommand{\newblock}{\relax}
\providecommand{\bibinfo}[2]{#2}
\providecommand{\BIBentrySTDinterwordspacing}{\spaceskip=0pt\relax}
\providecommand{\BIBentryALTinterwordstretchfactor}{4}
\providecommand{\BIBentryALTinterwordspacing}{\spaceskip=\fontdimen2\font plus
\BIBentryALTinterwordstretchfactor\fontdimen3\font minus
  \fontdimen4\font\relax}
\providecommand{\BIBforeignlanguage}[2]{{%
\expandafter\ifx\csname l@#1\endcsname\relax
\typeout{** WARNING: IEEEtran.bst: No hyphenation pattern has been}%
\typeout{** loaded for the language `#1'. Using the pattern for}%
\typeout{** the default language instead.}%
\else
\language=\csname l@#1\endcsname
\fi
#2}}
\providecommand{\BIBdecl}{\relax}
\BIBdecl

\bibitem{YChenComMag}
Y.~Chen, S.~Zhang, S.~Xu, and G.~Y. Li, ``Fundamental tradeoffs on green
  wireless networks,'' \emph{{IEEE Communications Magazine.}}, vol.~49, no.~6,
  pp. 30--37, June 2011.

\bibitem{DPC}
G.~Caire and S.~S. (Shitz), ``{On the Achievable Throughput of a Multiantenna
  Gaussian Broadcast Channel},'' \emph{IEEE Trans. Inf. Theory}, vol.~49,
  no.~7, pp. 1691--1706, Jul. 2003.

\bibitem{Scaling_TSDPC}
M.~Sharif and B.~Hassibi, ``A comparison of time-sharing, {DPC} and beamforming
  for {MIMO} broadcasting channels with many users,'' \emph{IEEE Trans.
  Commun.}, vol.~55, no.~1, pp. 11--15, Jan. 2007.

\bibitem{Limits_mumimo}
B.~Hassibi and M.~Sharif, ``Fundamental limits in {MIMO} broadcast channels,''
  \emph{IEEE J. Select. Areas Commun.}, vol.~25, no.~7, pp. 1333--1344, Sep.
  2007.

\bibitem{ZFDPC}
J.~Jiang, R.~Buehrer, and W.~Tranter, ``Greedy scheduling performance for a
  zero-forcing dirty-paper coded system,'' \emph{IEEE Trans. Commun.}, vol.~54,
  no.~5, pp. 789--793, May. 2006.

\bibitem{ZFBFSUS}
T.~Yoo and A.~Glodsmith, ``On the optimality of multiantenna broadcast using
  zero-forcing beamforming,'' \emph{IEEE J. Select. Areas Commun.}, vol.~24,
  no.~3, pp. 528--541, Mar. 2006.

\bibitem{Scaling_RBF}
M.~Sharif and B.~Hassibi, ``On the capacity of {MIMO} broadcast channels with
  partial side information,'' \emph{IEEE Trans. Inform. Theory}, vol.~51,
  no.~2, pp. 506--522, Feb. 2005.

\bibitem{fractionalprogramming}
\BIBentryALTinterwordspacing
C.~Isheden, Z.~Chong, E.~Jorswieck, and G.~Fettweis, ``{Framework for
  Link-Level Energy Efficiency Optimization with Informed Transmitter},''
  \emph{IEEE Trans. Wireless Commun.}, submitted. [Online]. Available:
  \url{http://arxiv.org/abs/1110.1990}
\BIBentrySTDinterwordspacing

\bibitem{EEWaterfillingDPC}
\BIBentryALTinterwordspacing
J.~Xu, L.~Qiu, and S.~Zhang, ``{Energy Efficient Iterative Waterfilling for the
  MIMO Broadcasting Channels},'' in \emph{IEEE proc. of WCNC 2012}, accepted.
  [Online]. Available: \url{http://home.ustc.edu.cn/~suming/}
\BIBentrySTDinterwordspacing

\bibitem{LambertWfunction}
R.~Corless, G.~Gonnet, D.~Hare, D.~Jeffrey, and D.~Knuth, ``On the {Lambert
  $W$} function,'' \emph{Advances in Computational Mathematics}, 1996.

\bibitem{Duality}
S.~Vishwanath, N.~Jindal, and A.~Goldsmith, ``{Duality, achievable rates, and
  sum-rate capacity of MIMO broadcast channels},'' \emph{IEEE Trans. Inf.
  Theory}, vol.~49, no.~10, pp. 2658--2668, Oct. 2003.

\bibitem{Arnold}
O.~Arnold, F.~Richter, G.~Fettweis, and O.~Blume, ``{Power Consumption Modeling
  of Different Base Station Types in Heterogeneous Cellular Networks},'' in
  \emph{Proceedings of the ICT MobileSummit (ICT Summit'10), Florence, Italy}.

\bibitem{Xu1}
J.~Xu, L.~Qiu, and C.~Yu, ``{Improving Energy Efficiency Through Multimode
  Transmission in the Downlink} {MIMO} {S}ystems,'' \emph{EURASIP J. Wireless
  Commun. and Net.}, vol. 2011, no.~1, p. 200, 2011.

\bibitem{frac_program}
S.~Schaible, ``{Fractional programming},'' \emph{Zeitschrift f\'ur Operations
  Research}, vol.~27, no.~1, pp. 39--54.

\end{thebibliography}

\end{document}